%% file: main_v3.tex
\newif\ifonecol 
\onecolfalse 
\ifonecol

\else
\documentclass[letterpaper, journal]{IEEEtran}
\fi

\newif\ifsavespace
\savespacetrue

\ifonecol

\else

\fi
\usepackage{amsthm,amsmath,amssymb,color,graphicx,caption,subcaption,tikz,url,latexsym,xfrac,bm} 

\usepackage{verbatim}
\usetikzlibrary{decorations.markings}
\usetikzlibrary{shapes.geometric}
\usepackage{cite}
\usepackage{enumerate}
\usepackage[draft,pdfborder={0 0 0}]{hyperref}

\usepackage{colortbl}
\usepackage{pgfplots}
\usepgfplotslibrary{units}
\usetikzlibrary{spy,backgrounds}
\usepackage{pgfplotstable}
\usepackage{flushend}
\usepackage[normalem]{ulem}
\usepackage{multirow}
\usepackage{glossaries} 
\usepackage[commandnameprefix=ifneeded,commentmarkup=uwave,final]{changes}
\definechangesauthor[name=Pascal, color=blue]{P}
\definechangesauthor[name=Ilshat, color=purple]{I}
\definechangesauthor[name=Charles, color=red]{C}
\definechangesauthor[name=Valerio, color=orange]{V}

\newtheorem{theorem}{Theorem}\theoremstyle{definition}

\newtheorem{proposition}{Proposition}

\newtheorem{definition}{Definition}

\newcommand{\Abf}{\mathbf{A}}
\newcommand{\ubf}{\mathbf{u}}
\newcommand{\xbf}{\mathbf{x}}
\newcommand{\bbf}{\mathbf{b}}
\newcommand{\ybf}{\mathbf{y}}
\newcommand{\GN}{\mathbf{G}_N}

\newcommand{\Fbb}{\mathbb{F}}
\newcommand{\Acal}{\mathcal{A}} 
\newcommand{\Ccal}{\mathcal{C}} 
\newcommand{\Ccalm}{\mathcal{C}_{\text{M}}} 
\newcommand{\Ccalmb}{\mathcal{C}_{\text{M}}^{\text{block}}} 
\newcommand{\Ccalmbr}{\mathcal{C}_{\text{M}}^{\text{BR}}} 
\newcommand{\Ccalb}{\mathcal{C}^{\text{block}}} 
\newcommand{\Ccalbr}{\mathcal{C}^{\text{BR}}} 
\newcommand{\Icalb}{\mathcal{I}_{\text{block}}} 
\newcommand{\Icalbr}{\mathcal{I}_{\text{BR}}} 
\newcommand{\Zcalb}{\mathcal{Z}_{\text{block}}} 
\newcommand{\Zcalbr}{\mathcal{Z}_{\text{BR}}} 
\newcommand{\Ical}{\mathcal{I}} 
\newcommand{\Fcal}{\mathcal{F}} 
\newcommand{\Zcal}{\mathcal{Z}} 

\newcommand{\BLTA}{\mathsf{BLTA}}

\newcommand{\rrn}{\mathcal{R}(r,n)}
\newcommand{\LT}{\text{LT}}
\newcommand{\NBstar}{\lvert\star\rvert}

\newcommand{\Ebb}{\mathbb{E}}
\DeclareMathOperator{\dec}{dec}
\DeclareMathOperator{\adec}{adec}

\newcolumntype{C}[1]{>{\centering\let\newline\\\arraybackslash\hspace{0pt}}m{#1}}

\begin{document}
\definecolor{darkred}{rgb}{0.75, 0, 0}
\definecolor{darkgreen}{rgb}{0, 0.45, 0}
\definecolor{bluegreen}{rgb}{0, 0.84, 0.84}
\definecolor{matlab1}{rgb}{0, 0.447, 0.741} 
\definecolor{matlab2}{rgb}{0.850, 0.325, 0.098} 
\definecolor{matlab3}{rgb}{0.929, 0.694, 0.125} 
\definecolor{matlab4}{rgb}{0.494, 0.184, 0.556} 
\definecolor{matlab5}{rgb}{0.466, 0.674, 0.188} 
\definecolor{matlab6}{rgb}{0.301, 0.745, 0.933} 
\definecolor{matlab7}{rgb}{0.635, 0.078, 0.184} 
\newcommand{\mypoint}[2]{\tikz[remember picture]{\node[inner sep=0, anchor=base](#1){$#2$};}}
\usetikzlibrary{positioning}
\newacronym{upo}{UPO}{universal partial order}
\newacronym{ae}{AE}{Automorphism Ensemble}
\newacronym{br}{BR}{bit-reversal}
\newacronym{syspc}{SP}{systematic polar}
\newacronym{spc}{SPC}{shortened polar}
\newacronym{scl}{SCL}{SC-List}
\newacronym{sc}{SC}{Successive Cancellation}
\newacronym{blta}{BLTA}{block-lower-triangular affine}
\newacronym{ber}{BER}{Bit-Error-Rate}
\newacronym{bler}{BLER}{block-error rate}
\newacronym{lt}{LT}{lower-triangular}
\newacronym{lta}{LTA}{lower-triangular affine}
\newacronym{llr}{LLR}{log-likelihood ratio}
\newacronym{ml}{ML}{maximum likelihood}
\newacronym{rm}{RM}{Reed-Muller}
\newacronym{bp}{BP}{Belief Propagation}
\newacronym{scan}{SCAN}{Soft CANcellation}
\newacronym{awgn}{AWGN}{Additive White Gaussian Noise}
\newacronym{bpsk}{BPSK}{Binary Phase-Shift Keying}
\newacronym{ca}{CA}{CRC-aided}
\newacronym{crc}{CRC}{cyclic redundancy check}
\newacronym{db}{dB}{decibel}
\newacronym{snr}{SNR}{signal-to-noise ratio}
\newacronym{urllc}{URLLC}{ultra-reliable low-latency communications}
\title{Shortened Polar  Codes under Automorphism Ensemble Decoding}
\author{Charles Pillet, Ilshat Sagitov, Valerio Bioglio, and Pascal Giard
\thanks{Charles Pillet, Ilshat Sagitov, and Pascal Giard (\{charles.pillet, ilshat.sagitov\}@lacime.etsmtl.ca, pascal.giard@etsmtl.ca) are with LaCIME, École de technologie supérieure (ÉTS), Montréal, QC, Canada. Valerio Bioglio (valerio.bioglio@unito.it) is with Department of Computer Science, Università degli Studi di Torino, Torino, Piedmont, Italy.}}

\maketitle
 
 \begin{abstract}
    In this paper, we propose a low-latency decoding solution of shortened polar codes based on their automorphism groups.
    The automorphism group of shortened polar codes, designed according to two existing shortening patterns, are shown to be limited but non-empty, making the \gls{ae} decoding of shortened polar codes possible. 
    Extensive simulation results for shortened polar codes under \gls{ae} are provided and are compared to the \gls{scl} algorithm.
    The block-error rate of shortened polar codes under \gls{ae} matches or beats \gls{scl} while lowering the decoding latency. 
\end{abstract}
\glsresetall
	\IEEEpeerreviewmaketitle	
	\section{Introduction}
    Polar codes \cite{ArikanPolarCodes} are based on the channel polarization induced by the matrix $\GN$  recursively generated from the binary kernel $\mathbf{G}_2$. 
    Polar codes achieve the capacity of symmetric binary memoryless channels under the low-complexity \gls{sc} decoding algorithm \cite{ArikanPolarCodes}.
    However, the capacity is achieved asymptotically on the code length. Polar codes exhibit a poor maximum likelihood (ML) bound in the finite length regime.
    Constructions of the information set reducing the polarization have been proposed to improve the finite length performance \cite{FromPolartoRM}.
    Regardless of the information set, the designed code is a $\GN$-coset code \cite{ArikanPolarCodes}, a code based on the matrix $\GN$.
    For a finite length $N$, concatenating a \gls{crc} to the polar code permits to improve the error-correction performance. 
    \Gls{ca}-polar codes are usually decoded with the \gls{scl} \cite{SCL} algorithm.
    \Gls{ca}-polar codes are included in 5G-New Radio as one of the coding schemes for the control channel \cite{5GHuawei}.

    During the standardisation, both academia and industry joined their efforts to make the polar code parameters more versatile.  
    Various length-matching techniques have been successfully applied to polar codes to adjust the code length with 1-bit granularity, including shortening, puncturing and extending \cite{5GHuawei,shortening_dynamic,Punc_short_PC_nonsys,Punc_but_short}.
    In a shortened code, $S$ bits known at the decoder are not transmitted.
    In practice, these bits are included in the frozen set.
    In the 5G polar-code shortening scheme \cite{5GHuawei}, called block shortening \cite{Punc_but_short}, the last $S$ bits are not transmitted. 
    The most reliable bit locations being at the end, block shortening discards many highly-reliable locations.
    To overcome this issue, a \gls{br} shortening pattern composed by the last $S$ bits indexed by the \gls{br} permutation has been proposed in \cite{Punc_short_PC_nonsys}. 
    
	\Gls{ae} decoding \cite{geiselhart2020automorphism} performs similarly to the 5G polar codes using CA-SCL  \cite{RateCompatibleAE}.
    In \gls{ae} decoding, $M$ permutations are selected in the affine automorphism group of the code to generate $M$ instances of the noisy vector.
	The $M$ permuted noisy vectors are decoded in parallel and the most likely candidate codeword is selected.
    \Gls{ae} decoding exhibits good decoding performance with ``structured'' $\GN$-coset codes \cite{PSMC}. 
    These codes are compliant with the \gls{upo}, an order among the integers valid for all channel conditions.
    Codes compliant with the \gls{upo}, denoted as \emph{polar-like} in the paper, exhibit a large automorphism group \cite{geiselhart2021automorphismPC}, useful for \gls{ae} decoding, but require specific designs \cite{AE_v2}.
    These constraints make one-bit granularity rate-compatible sequences of such codes impossible \cite{RateCompatibleAE}.
	\Gls{ae} decoding was also proposed for polar codes in \cite{geiselhart2021automorphismPC,AE_v1}.
    Recently, a hardware implementation of \gls{ae} \cite{AESC_HW} has been shown to outperform \gls{scl} decoders in latency, area and energy efficiency while offering similar error-correction performance for short codes \cite{RateCompatibleAE}. 
            
    Modern wireless communications need rate-matching techniques \cite{5GHuawei} and low-latency decoding solutions. 
    In this paper, we investigate low-latency \gls{ae} decoding of shortened polar codes.    
    The shortening constraint destroys the code structure leading to an uncertain automorphism group.
    The automorphism group is found by considering the shortening constraint while describing these codes as $\GN$-coset codes \cite{AE_v1}.
    A thorough analysis of the computed automorphism groups is performed for all shortening parameters and shows that the \gls{ae} algorithm can  decode any shortened polar code.
    Simulation results in terms of \gls{bler} and average decoding time illustrate the benefits of the approach.
\section{GN-coset codes definition and decoding}\label{sec:preliminaries}
The binary kernel $\mathbf{G}_2=\left[\begin{smallmatrix}1 &0\\ 1 &1\end{smallmatrix}\right]\in\Fbb_2^{2\times 2}$, where $\Fbb_2$ denotes the binary field, generates the transformation matrix $\GN=\mathbf{G}_2^{\otimes n}$, where $(\cdot)^{\otimes n}$ denotes the $n^{th}$ Kronecker power.
The \emph{$\GN$-coset codes} are based on the transformation matrix $\GN$ \cite{ArikanPolarCodes}.
\begin{definition}[$\GN$-coset codes]
A $(N,K)$ $\GN$-coset code of length $N$ and dimension $K$ is defined by the information set $\Ical\subseteq[N]\triangleq\{0,1,\dots,N-1\}$, with $|\Ical|=K$. The frozen set is $\Fcal=\Ical^c\subseteq[N]$. 
Encoding is performed as $\xbf=\ubf \GN$, where $\ubf,\xbf\in\Fbb_2^N$ are the input vector and the codeword, respectively.
The input vector $\ubf$ is generated by assigning $u_i=0$ if $i\in\Fcal$ and by storing $K$ information bits in the positions stated in $\Ical$.
\end{definition}
The number of possible $(N,K)$ $\GN$-coset codes is very large.  
Polar codes, Reed-Muller codes, and polar-like codes are all defined as sub-families of $\GN$-coset codes.
The connections among these sub-families are depicted in \autoref{fig:subfamily}.
\begin{figure}
    \centering
    \resizebox{0.30\textwidth}{!}{\input{figures/gn_coset_codes}}
    \caption{Representation of sub-families of $\GN$-coset codes.}
    \label{fig:subfamily}
\end{figure}
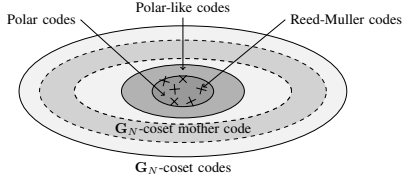

\subsection{$\GN$-Coset Sub-Families}
Polar codes are defined based on the polarization induced by $\GN$ \cite{ArikanPolarCodes}.
A channel $W$ with symmetric capacity $0\leq I(W)\leq1$ is split by $\mathbf{G}_2$ in two sub-channels $W^-$ and $W^+$ verifying $I(W^-)<I(W)<I(W^+)$.
Given the recursive construction of $\GN=\mathbf{G}_2^{\otimes n}$, $\GN$ produces asymptotically reliable ($I(W^+)=1$) and unreliable ($I(W^-)=0$) bit-channels. 
For finite code lengths, some of the bit-channels are not fully polarized leading to alternative designs of $\Ical$.
Polar codes use the polarization induced by $\GN$ to define their $\Ical$.
\begin{definition}[Polar codes]\label{def:polar_codes}
A $(N,K)$ polar code is a $(N,K)$ $\GN$-coset code for which $\Ical$ includes the $K$ most reliable bit-channels resulting from the polarization of a given channel. 
\end{definition}
\glsfirst{rm} codes \cite{Reed,Muller} can as well be seen as a sub-family of $\GN$-coset codes too \cite{ArikanPolarCodes}.
\begin{definition}[Reed-Muller codes]\label{def:RM_codes}
A $\rrn$ \gls{rm} code of order $r$ and variable $n$ is a $\GN$-coset code with parameters $\left(N=2^n,K=\sum_{k=0}^{r}{n \choose k}\right)$ for which $\Ical$ is composed by all integers $k$ whose binary representation Hamming weight is greater or equal to $n-r$.
\end{definition}
Contrary to polar codes, \gls{rm} codes are not channel dependent.
However, for a code length $N=2^n$, only $n+1$ code dimensions allow a \gls{rm} code construction.
A \emph{universal partial order} (\gls{upo}) among virtual channels exists and is calculated through two comparison rules \cite{partial_order}. 
The partial order $i\succeq j$ compares two indices $i,j\in[N]$, proving that the $j^{th}$ bit-channel is always weaker than the $i^{th}$, no matter the channel.
Every couple of indices are not comparable, hence the name ``partial''. 
The comparison rules are not given, but the last index $N-1$ is always better no matter the channel.
\emph{Polar-like codes} are based on the \gls{upo} and include polar and \gls{rm} codes as shown in \autoref{fig:subfamily}.
\begin{definition}[Polar-like codes]\label{def:polar_like_codes}
A $(N,K)$ polar-like code is a $(N,K)$ $\GN$-coset code compliant with the \gls{upo}, i.e, $i\succeq j \, \text{and} \,j\in\Ical \implies i\in\Ical.$
\end{definition}
\subsection{Affine Automorphism Group of $\GN$-Coset Codes}
An automorphism $\pi$ of a code $\Ccal$ is a permutation mapping any codeword into another codeword of the codebook. 
The affine automorphism group $\mathcal{A}(\mathcal{C})$ is the set of automorphisms that can be described as an affine transformation, i.e., $\mathbf{z}' = \Abf\mathbf{z} + \bbf$ where $\mathbf{z},\mathbf{z}'\in\Fbb_2^n$ are the binary representation of an integer $i\in[N]$, $\mathbf{A}\in\Fbb_2^{n\times n}$ is an invertible matrix and $\bbf\in\Fbb_2^n$ a vector.
For polar-like codes, $\mathcal{A}(\mathcal{C})$ is isomorphic to the \gls{blta} group $\BLTA(\mathbf{S})$ with profile $\mathbf{S}$, denoting the size of the blocks alongside the diagonal \cite{geiselhart2021automorphismPC,AEjournal}.
For \gls{rm} codes, $\Acal(\Ccal)$ is isomorphic to the general affine group \cite{WilliamsSloane}, i.e., $\mathbf{S}=n$.
If the code is not polar-like, $\Acal(\Ccal)$ likely exhibits an irregular shape \cite{AE_v1}.

\subsection{Decoders of $\GN$-Coset Codes}
\glsfirst{sc} is a soft-input/hard-output decoder that can be described as a binary tree search, where the tree is traversed depth-first starting from the left branch \cite{ArikanPolarCodes}.
Polar codes are defined to be capacity-achieving under \gls{sc} decoding over binary-memoryless channels \cite{ArikanPolarCodes}.
\glsfirst{scl} is a list decoder based on \gls{sc} \cite{SCL}, namely, a decoder running up to $L$ SC decoding instances in parallel to follow the $L$ best decoding paths.
\gls{scl} offers an outstanding decoding performance for CRC-aided (\gls{ca}) $\GN$-coset codes.
\gls{bp}  and \gls{scan}  are two soft-input/soft-output iterative decoding algorithms of $\GN$-coset codes. 
\gls{scan} iterates using the \gls{sc} schedule while \gls{bp} iterates by passing from one stage of the factor graph to another leading to higher throughput \cite{geiselhart_subcode}.
Due to their iterative nature, SCAN and BP are generally more complex than SC.

Given an aforementioned decoder $\dec$, the automorphism decoder $\adec$ \cite{AE_v2} is defined as:
\begin{equation}
\label{eq:adec}
\adec(\ybf,\pi, \Fcal) = \pi^{-1}\left( \dec(\pi(\ybf),\Fcal) \right),
\end{equation}
where $\ybf\in\mathbb{R}^N$ is the received signal and $\pi \in \Acal(\Ccal)$. 
An \gls{ae} decoder \cite{geiselhart2020automorphism} consists of $M$ $\adec$ instances running in parallel with different $\pi_m\in\mathcal{A}(\mathcal{C})$ for $ 1\leq m\leq M$. Each $\adec$ returns a codeword candidate $\hat{\xbf}_m$.
A \gls{ml}-criteria is used to select the codeword candidate $\hat{\xbf}$.

\subsection{Code Shortening}\label{subsec:shortening}
A $\GN$-coset mother code $\Ccalm$ of length $N$ can be shortened to length $N-S$ if there exists a shortening set $\Zcal\subset[N]$, with $1\leq |\Zcal| = S < \frac{N}{2}$, such that $\xbf_\Zcal=\mathbf{0}$ regardless the values of the information bits. 
These shortened bits $\xbf_\Zcal$ are not transmitted and are evaluated as infinite \glspl{llr} during the decoding.
The information set $\Ical$ is designed selecting $K$ indices from $[N]\setminus\mathcal{Z}$. 
If the $K$ most reliable virtual channels are chosen, a \emph{shortened polar code} is obtained:
\begin{definition}[Shortened polar codes]\label{def:SPC}
    A $(N',K',\Zcal)$ shortened polar code with shortening set $\Zcal$ is a code based on a $(N=N'+S,K=K')$ $\GN$-coset mother code $\Ccalm$ such that $\Zcal\subseteq\Fcal$ and $\Ical$ includes the $K$ most reliable virtual channels in $[N]\setminus\Zcal$. 
\end{definition}
Two strategies of designing shortening sets have been proposed to force $\xbf_\Zcal=\mathbf{0}$; either introducing dynamic frozen bits \cite{shortening_dynamic} or finding a frozen bits subset that is independent from $\Ical$.
The latter corresponds to \emph{block shortening} (if $\Zcal$ is composed of the last $S$ bits) \cite{Punc_but_short} and the \emph{bit-reversal (BR) shortening} (if $\Zcal$ is composed of the last $S$ bits from the \gls{br} permutation) \cite{Punc_short_PC_nonsys}.
In this paper, both block and \gls{br} patterns are studied.

Definition \ref{def:SPC} can be modified by having  $\Zcal\subset\Ical$ leading to $|\Ical|=K'+S$.
It still requires $\xbf_\Zcal=\mathbf{0}$.
However, the feedback during the decoding is expected to be inferior, since there are less bits considered frozen by the decoder.
For this reason, the shortening approach $\Zcal\subseteq\Fcal$ will be studied by default while the approach $\Zcal\subset\Ical$ will be clearly stated if used.

\section{Group Properties of Shortened Polar Codes}\label{sec:study_nonsys}
The affine automorphism group of the mother code is studied as its knowledge is required to decode shortened polar codes with \gls{ae} \eqref{eq:adec}.
The action of the permutations on the mother code is also investigated.
\subsection{Polar-like Codes and $\GN$-Coset Mother Codes}\label{subsec:exampleshortened}
Block and \gls{br} patterns can be used to guarantee that $\xbf_\Zcal=\mathbf{0}$ \cite{Punc_short_PC_nonsys}, i.e., to perform shortening. 
It leads to Proposition \ref{prop:short_polarlike}.
\begin{proposition}\label{prop:short_polarlike}
The $\GN$-coset mother code $\Ccalm$ of any block or BR shortened polar code is not a polar-like code.
\end{proposition}
\begin{proof}
Given the definition of block and \gls{br} shortening, the index $N-1\in[N]$ is the first to be shortened, thus $N-1\in\Zcal\subseteq\Fcal$.
Since, $\forall j\in[N-1], N-1\succeq j$ and $N-1\notin\Ical$, $\Ccalm$ is not compliant with the \gls{upo} and thus is not a polar-like code according to Definition~\ref{def:polar_like_codes}.
\end{proof}
Given Proposition~\ref{prop:short_polarlike}, the $\GN$-coset mother codes are located  in the $\GN$-coset code family but outside the polar-like code sub-family in \autoref{fig:subfamily}.
In order to decode a shortened polar code with the \gls{ae} decoder, the knowledge of the affine automorphism group $\Acal(\Ccalm)$ of its $\GN$-coset mother code $\Ccalm$ is required.
However, since $\Ccalm$ is not compliant with the \gls{upo}, its affine automorphism group is likely to not be isomorphic to one of the \gls{blta} groups \cite{AE_v1}.
Next, we focus on the affine automorphism group $\Acal(\Ccalm)$.
\subsection{Affine Automorphism Group of Mother Codes}\label{subsec:aff_auto_group_SPC}
The affine automorphism group $\Acal(\Ccal)$ of any $\GN$-coset code $\Ccal$ is retrieved on the basis of the information set $\Ical$ \cite[Theorem 2]{AE_v1}. 
It returns a list of admissible positions, noted $\star$ in the binary invertible matrix $\mathbf{A}$. 
Both values $\{0,1\}$ are allowed on a $\star$'s position.
The set of permutations given by \cite[Theorem 2]{AE_v1} for a given $\GN$-coset code is noted $\Pi$.     
For the mother code $\Ccalm$ of a shortened polar code, $\Pi$ is refined to retrieve $\Acal(\Ccalm)$.
\begin{theorem}\label{theo:group_shortening}
    The affine automorphism group $\Acal(\Ccalm)$ of a $\GN$-coset mother code $\Ccalm$ corresponds to the permutations $\pi\in\Pi$ given by \cite[Theorem 2]{AE_v1} verifying 
    \begin{align}\label{eq:constraint_perm}
        \pi(\Zcal)=\Zcal,
    \end{align}
    where $\pi(\Zcal)=\{\pi(z)\,:\,z\in\Zcal\}$.
\end{theorem}
\begin{proof}
    \cite[Theorem 2]{AE_v1} gives the positions in $\mathbf{A}$ that do not modify the frozen set and information set, such that $\Acal(\Ccalm)\leq\Pi$.
    The shortening pattern ensures that the $S$ coded bits in $\Zcal$ are linear combinations of frozen bits only, i.e.,
    \begin{align}\label{eq:constraint_shortening}
        \GN(\Ical,\Zcal)=\mathbf{0},
    \end{align}
    where $\GN(\Ical,\Zcal)$ is the sub-matrix of $\GN$ formed by the rows stated in $\Ical$ and the columns stated in $\Zcal$ \cite{Punc_short_PC_nonsys}.
    Applying a permutation $\pi$ to a $\GN$-coset code is to change  $\GN$ by  $\mathbf{G}_T=\GN\mathbf{T}\GN$ \cite{geiselhart_subcode} where $\mathbf{T}\in\Fbb^{N\times N}$ is the permutation matrix based on $\pi$. 
    If $\pi\in\Pi$ and \eqref{eq:constraint_perm} is verified then the $S$ coded bits remain a linear combination of frozen bits, i.e.,
    \begin{align}\label{eq:transformed_constraint}
        \mathbf{G}_T(\Ical,\Zcal)=\mathbf{0}.
    \end{align}
    Given \eqref{eq:transformed_constraint}, the constraint \eqref{eq:constraint_shortening} remains valid such that $\pi\in\Acal(\Ccalm)$.    
    However, if $\pi\in\Pi$ and \eqref{eq:constraint_perm} is not verified, we have 
    \begin{align}\label{eq:constraint_not_verified}
        \mathbf{G}_T(\Ical,\Zcal)\neq\mathbf{0},
    \end{align}
    the shortening constraint \eqref{eq:constraint_shortening} is not verified.
    Hence, we have $\pi\notin\Acal(\Ccalm)$, finishing the proof. 
\end{proof}
As described in Theorem \ref{theo:group_shortening}, automorphisms of a mother code $\Ccalm$ represent only a portion of the permutations in the computed $\Pi$.
As an example, the $\Ccalbr(12,3,\Zcalbr)$ \gls{br} and $\Ccalb(12,3,\Zcalb)$ block shortened polar codes are studied. 
We have $\Zcalbr=\{3,7,11,15\}$ and $\Icalbr=\{12,13,14\}$ while $\Zcalb=\{12,13,14,15\}$ and $\Icalb=\{7,10,11\}$.
By using $\Icalbr$ and $\Icalb$ \cite[Theorem 2]{AE_v1}, the affine transformations not changing the structure of the $\GN$-coset mother codes $\Ccalmbr$ and $\Ccalmb$ are given by 	$\Abf_{\text{BR}}=\left[\begin{smallmatrix}
  		\star & \star & 0 & 0 \\
  		\star & \star & 0 & 0 \\
  		\star & \star & \star & \star \\
  		\star & \star & \star & \star
  	\end{smallmatrix}\right]$ and $\Abf_{\text{block}}=\left[
    \begin{smallmatrix}
  		1 & 0 & \star & 0 \\
  		\star & 1 & \star & \star \\
  		0 & 0 & 1 & 0 \\
  		0 & 0 & \star & 1
  	\end{smallmatrix}\right]$.
$\Abf_{\text{BR}}$ allows $\NBstar=12$ admissible positions  and we note that $\Acal(\Ccalmbr)\leq\BLTA(2,2)=\Pi$, the size of such group is known and is
 $|\BLTA(2,2)|=9216$ \cite{geiselhart2021automorphismPC,AE_v2}.
By exhaustive search, we found $|\Acal(\Ccalmbr)|=2304$ automorphisms, i.e., permutations $\pi\in\BLTA(2,2)$ verifying \eqref{eq:constraint_perm}.
The structure of the invertible matrix $\Abf_{\text{block}}$ limits the number of invertible matrices, and thus, the number of permutations $|\Pi|$.
We computed $|\Pi|=512$ and $|\Acal(\Ccalmb)|=128$ automorphisms verifying \eqref{eq:constraint_perm} and forming $\Acal(\Ccalmb)$.

The ratio $\frac{|\Acal(\Ccalm)|}{|\Pi|}$ is studied for $N=256$.
Given $|\Zcal|=S$ and a randomly generated permutation $\pi\in\Pi$, $E_S$ denotes the event $\pi(\Zcal)=\Zcal$,i.e., $\pi\in\Acal(\Ccalm)$. Its expected value is noted $\Ebb(E_S)$. 
A low $\Ebb(E_S)$ translates into a smaller order of $\Acal(\Ccalm)$, noted $|\Acal(\Ccalm)|$.
For $K=120$, $N=256$ and $1\leq S\leq \frac{N}{2}$, $\Ebb(E_S)$ is going down to $\Ebb(E_{97})=0.23\%$.
The greater probability is obtained for $S=64$ with $\Ebb(E_{64})=20.2\%$.
The mean is $\frac{1}{127}\sum_{S=1}^{127}\Ebb(E_S)=1.48\%$.
For the block pattern, the probabilities are even lower.
The mean is $\frac{1}{127}\sum_{S=1}^{127}\Ebb(E_S)=0.48\%$.
Hence, finding automorphisms for \gls{ae} decoding is tricky and requires pre-computations based on $\Ccalm$, $\Pi$, and $\Zcal$. 
Next, we investigate $\Acal(\Ccalm)$ for all parameters $(S,K')$.

\subsection{Analysis For All Shortening Parameters $(S,K')$}
Theorem~\ref{theo:group_shortening} states that $\Acal(\Ccalm)\leq \Pi$. 
Since $|\Pi|$ grows exponentially with $N$, finding $|\Acal(\Ccalm)|$ requires an impracticable number of simulations, especially since $\Pi$ and $\Acal(\Ccalm)$ change for all pairs $(S,K')$. 
To estimate the relative size of $\Pi$, the number of admissible positions $\NBstar$ in $\Abf$ is counted. 
The analysis is for both \gls{br} and block patterns and all shortening parameters $(S,K')$, namely, $0 < S <\frac{N}{2}$ and $0\leq K'\leq N-S$.
For a pair $(S,K')$, the shortened codes are designed at a \gls{snr} of $0$ \gls{db}. 
The number of admissible positions $\NBstar$ is compared with the number of admissible positions available $\NBstar_{\LT}$ for the smaller affine automorphism group $\Acal(\Ccal)$ of a polar-like code, i.e., the \gls{lta} group. 
A lower-triangular matrix $\mathbf{L}\in\Fbb_2^{n\times n}$ has $\NBstar_{\LT}=\sum_{i=1}^{n-1}i=\frac{n(n-1)}{2}$ admissible positions.
$\NBstar_{\LT}$ is used as a threshold to classify the $\GN$-coset mother codes $\Ccalm$  in two sets of codes $\Ccal_{\text{large}}$ and $\Ccal_{\text{small}}$.
If $\NBstar>\NBstar_{\LT}$ $\left(\NBstar\leq\NBstar_{\LT}\right)$, $\Ccalm$ is considered to have a large (small) affine automorphism group $\Acal(\Ccalm)$ and belongs to $\Ccal_{\text{large}}$ ($\Ccal_{\text{small}}$). 
\begin{figure}[t]
\centering    
\includegraphics[width=0.75\columnwidth]{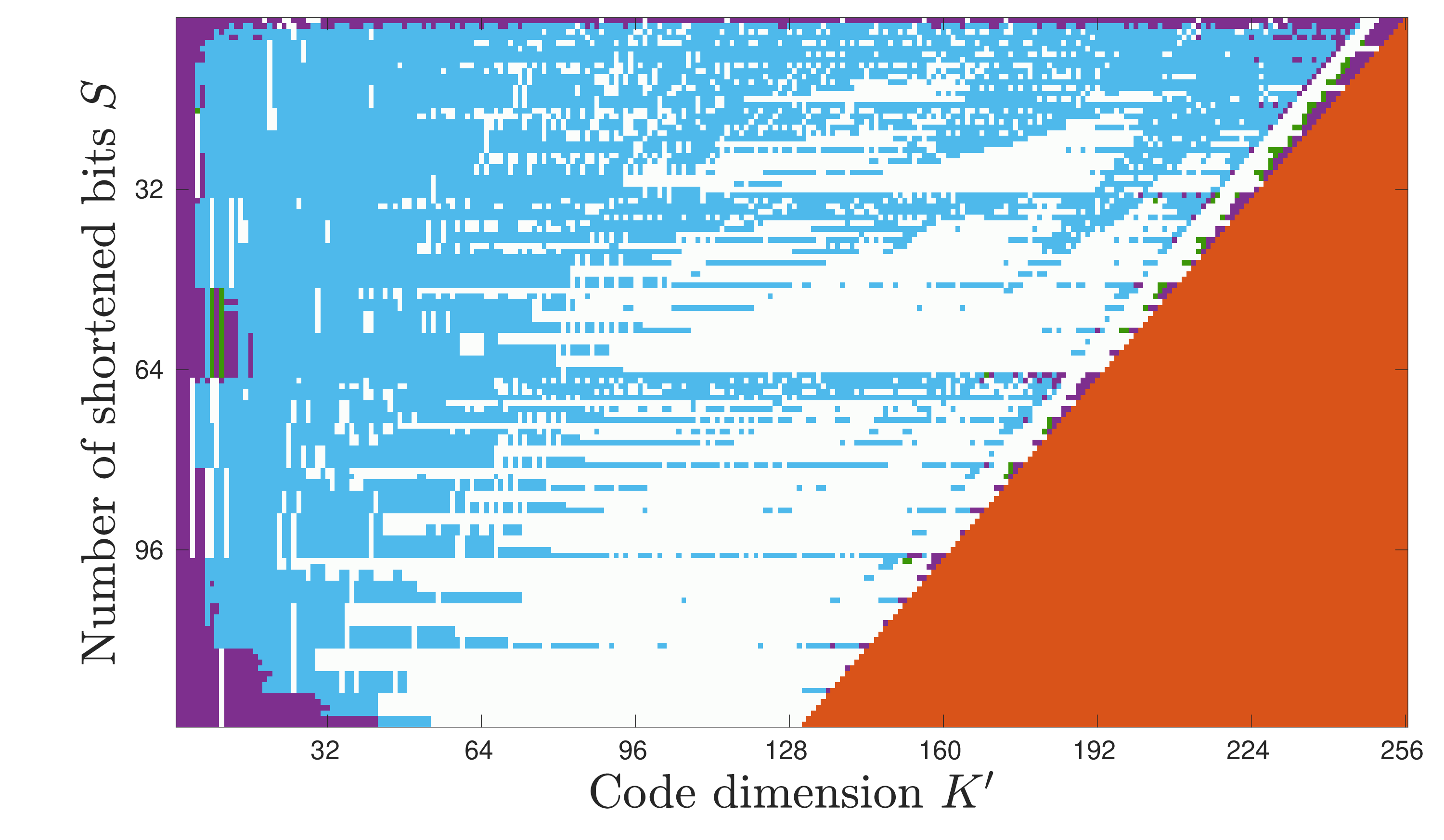}
    \caption{Classification of $\Ccalmbr$ and $\Ccalmb$ for $N=256$ and a design \gls{snr} of $0$\,dB.
    Blue (green) positions correspond to $\Ccalmbr\in\Ccal_{\text{large}}$ $\left(\Ccalmb\in\Ccal_{\text{large}}\right)$, purple positions to  $\Ccalmbr\wedge\Ccalmb\in\Ccal_{\text{large}}$, and white positions to  $\Ccalmbr\wedge\Ccalmb\notin\Ccal_{\text{large}}$.}
    \label{fig:256_combination}
\end{figure}
The number of codes with $\NBstar>\NBstar_{\LT}$ is denoted $|\Ccal_{\text{large}}|$ and is compared to the number of possible combinations $(S,K')$, denoted $|(S,K')|=\sum_{s=1}^{N/2-1}N-s=\frac{3N}{4}\left(\frac{N}{2}-1\right)$.
As for the example in Section\,\ref{subsec:aff_auto_group_SPC}, the \gls{br} pattern is more likely to generate $\Ccalm\in\Ccal_{\text{large}}$ compared to the block pattern.
For $N=\{64,128,256\}$, more than half of the $\GN$-coset mother codes for the \gls{br} pattern belong to $\Ccal_{\text{large}}$. 
For both patterns, the ratio $\frac{|\Ccal_{\text{large}}|}{|(S,K')|}$ is decreasing with the code length $N$ and increasing with the design \gls{snr}.

The analysis is completed by checking which pairs $(S,K')$ allow $\Ccalm\in\Ccal_{\text{large}}$ for \gls{br} $\left(\Ccalmbr\right)$ and block $\left(\Ccalmb\right)$ shortening.
\autoref{fig:256_combination} illustrates these classification for $N=256$.
For $K'\approx0$ or for $S\approx0$, both patterns exhibit $\NBstar>\NBstar_{\LT}$ and are depicted in purple.
The white positions represent shortened codes with small $\NBstar\leq\NBstar_{\LT}$ regardless of the shortening pattern. 
These codes are mostly located on a diagonal for which $K'+S>\frac{N}{2}$.
Finally positions in orange correspond to non-achievable shortening parameters, i.e., $K'>N-S$, namely the  message length is larger than the number of bits sent over the channel.


\subsection{Action of Permutations} \label{subsec:nonsys_perm_constraint}   
To conclude this section, the action of permutations on $\Ccalm$ is investigated.
Applying an automorphism $\pi\in\Acal(\Ccalm)$ to $\Ccalm$ does not generate dynamic frozen bits. 
However, the action of $\pi\in\Pi\setminus\Acal(\Ccalm)$ on $\Ccalm$ could, thus it is studied.
\begin{proposition}
    \label{prop:dyn_short}
    Applying a permutation $\pi\in\Pi\setminus\Acal(\Ccalm)$ on a $\GN$-coset mother code $\Ccalm$ transforms at least one of the frozen bits $i\in\Fcal\cap\Zcal$ into a dynamic frozen bit.
    All frozen bits in $\Fcal\setminus\Zcal$ remain frozen to 0.
\end{proposition}
\begin{proof}
    Given the definition of $\Pi$, applying a permutation $\pi\in\Pi$ does not change $\Ical$ and $\Fcal$ and the frozen bits in $\Fcal\setminus\Zcal$ remain frozen to 0 \cite{AE_v1}.
    However, $\Ccalm$ is now constrained due to shortening. 
    Applying such permutation induces \eqref{eq:constraint_not_verified}, hence one of the shortened bits $\xbf_\Zcal$ is a combination of information bits.
    To have $\xbf_\Zcal=\mathbf{0}$, $\mathbf{u}_\Zcal$ should cancel out this combination, finishing the proof.
\end{proof}
For example, the permutation $\pi=(0,3,14,\mathbf{13},8,11,6,\mathbf{5},4,7,10,\mathbf{9},12,15,2,\mathbf{1})\in\Pi\setminus\Acal(\Ccalm)$ is applied on the $(16,3)$ $\GN$-coset mother code of the $(12,3)$ \gls{br} shortened polar code $\Ccalbr$ of Section\,\ref{subsec:aff_auto_group_SPC}.
Proposition \ref{prop:dyn_short} states that the resulting freezing constraints contain at least one dynamic frozen bit since $\pi(\Zcalbr)=\{13,5,9,1\}\neq\Zcalbr$.
In this example, the shortening constraint is 
\begin{align}
    \mathbf{G}_T(\Icalbr,\Zcalbr) = \begin{bmatrix}
        0 & 0 & 0 & \mypoint{x12}{0}\\
        0&0&0&\mypoint{x13}{1}\\
        \mypoint{x3}{0}&\mypoint{x7}{0}&\mypoint{x11}{0}&\mypoint{x15}{0}    
    \end{bmatrix}\hspace{10mm}.
\end{align}
\begin{tikzpicture}[remember picture, overlay]
      \node[right=10mm of x12,anchor=east](textofhere2){$u_{12}$};
      \node[right=10mm of x13,anchor=east](textofhere2){$u_{13}$};
      \node[right=10mm of x15,anchor=east](textofhere2){$u_{14}$};
      \node[below=2pt of x3,rotate=-90,anchor=west](textofhere2){$x_3$};
      \node[below=2pt of x7,rotate=-90,anchor=west](textofhere2){$x_7$};
      \node[below=2pt of x11,rotate=-90,anchor=west](textofhere2){$x_{11}$};
      \node[below=2pt of x15,rotate=-90,anchor=west](textofhere2){$x_{15}$};
    \end{tikzpicture}
    \vspace{6pt}
    
Hence, $x_{15}\in\xbf_{\Zcalbr}$ carries the value of $u_{13}$. It leads to the dynamic freezing constraint $u_{15}=u_{13}$ in order to have $x_{15}=u_{13}+u_{15}=0$.
Next, \gls{ae} decoding using $\pi\in\Pi$ will be labeled as \emph{adjusted}. 
It requires additional memory with respect to \gls{ae} since it tracks the dynamic freezing constraints \cite{memory_AE_dynamic}.

\added[id=C]{The dynamic frozen bits are caused by the shortening constraint $\Zcal\subseteq\Fcal$.
As explained in Section \ref{subsec:shortening}, the shortening approach $\Zcal\subset\Ical$ for the mother code $\Ccalm$ is also possible.}
If this construction is chosen, bits in $\mathbf{u}_\Zcal$ are not considered as frozen during the computation of $\Acal(\Ccalm)$ \cite[Theorem 2]{AE_v1}.
Hence, the group of permutations $\Pi$ changes.
By definition of $\Pi$, applying a permutation $\pi\in\Pi$  does not change $\Ical$ and $\Fcal$ and the frozen bits in $\Fcal\setminus\Zcal=\Fcal$ remain frozen to 0 \cite{AE_v1}.
Hence, any permutations in $\Pi$ can be used, leading to $\Acal(\Ccalm)=\Pi$.
As a trade-off, the feedback during $\adec$ \eqref{eq:adec} is less reliable since $|\Fcal|$ is smaller.
    

\section{Simulation Results}\label{sec:sim_results}
The \gls{bler} performance of $(115,51,\Zcal)$ block and \gls{br} shortened polar codes  under \gls{ae} decoding are given.
A random selection of automorphisms was used.
The simulations are carried out over the \gls{awgn} channel with the \gls{bpsk} modulation.
\Gls{ae}-$M$-$\dec$ denotes \gls{ae} decoding using $\dec$ as sub-decoders with $M$ automorphisms.
\gls{scl} with a list size $L$ is denoted \gls{scl}-$L$. 
\gls{scan} and \gls{bp} using up to $T$ iterations are denoted \gls{scan}-$T$ and \gls{bp}-$T$, respectively.

\autoref{fig:AE_block} depicts the \gls{bler} performance of $(115,51,\Zcalb)$ code. 
SCL-4 ($\mathbin{
\tikz[baseline]{ \draw[color=black,-,thick] (0pt,.5ex) -- (2ex,.5ex) -- (4ex,.5ex);
\fill[color=black,-,thick] (1.5ex,.0ex) rectangle (2.5ex,1ex);}}$) and SC ($\mathbin{
\tikz[baseline]{ \draw[color=black,-,thick] (0pt,.5ex) -- (2ex,.5ex) -- (4ex,.5ex);
\draw[color=black,-,thick] (1.5ex,.0ex) rectangle (2.5ex,1ex);}}$) are shown for reference, as well as 5G's shortened polar code under CA-SCL-4 ($\mathbin{
\tikz[baseline]{ \draw[color=black!40,-,thick] (0pt,.5ex) -- (2ex,.5ex) -- (4ex,.5ex);
\fill[color=black!40,-,thick] (1.5ex,.0ex) rectangle (2.5ex,1ex);}}$).
\begin{align}
    \Abf_{\text{block}}=\left[
      \begin{smallmatrix}
  		\star & \star & 0 & 0& \star & \star & \star \\
  		\star & \star & 0 & 0 & \star & \star & \star\\
  		0 & 0 & 1 & 0 & \star & \star & \star\\
  		0 & 0 & 0 & 1& \star & \star & \star\\
        0 & 0 & 0 & 0& \star & \star & \star\\
        0 & 0 & 0 & 0& \star & \star & \star\\
        0 & 0 & 0 & 0& \star & \star & \star
  	\end{smallmatrix}\right]
   ,\,\,\,
   \Abf_{\text{BR}}=\left[
    \begin{smallmatrix}
        \star & \star & \star & 0& 0 & 0 & 0 \\
  		\star & \star & \star & 0 & 0 & 0 & 0\\
  		\star & \star & \star & 0 & 0 & 0 & 0\\
  		\star & \star & \star & 1 & 0 & 0 & 0\\
        \star & \star & \star & 0 & 1 & 0 & 0\\
        \star & \star & \star & 0 & 0 & \star & \star\\
        \star & \star & \star & 0 & 0 & \star & \star
  	\end{smallmatrix}\right]. \label{eq:matrixsim}
\end{align}
$\Abf_{\text{block}}$ \eqref{eq:matrixsim} generates $\Pi$ for the $(128,51)$ $\GN$-coset mother code, $\Abf_{\text{block}}$ has $\NBstar=25>15=\NBstar_{\LT}$.
Moreover, $\Abf_{\text{block}}$ exhibits $\star$'s location on upper right which is good performance-wise \cite{AE_v1}.
\gls{ae}-4-\gls{bp}-200 ($\mathbin{
\tikz[baseline]{\draw[color=orange,-,thick] (0pt,.5ex) -- (2ex,.5ex) -- (4ex,.5ex);
\draw[solid, color=orange,thick] (1.5ex,0ex) -- (2.5ex,1ex);
\draw[solid, color=orange,thick] (1.5ex,1ex) -- (2.5ex,0ex);}}$) and \gls{ae}-4-\gls{scan}-5 ($\mathbin{
\tikz[baseline] \draw[color=matlab1,-,thick] (0pt,.5ex) -- (2ex,.5ex) circle (.5ex) -- (4ex,.5ex);}$) outperform \gls{scl}-4  by 0.5\,dB but with a higher complexity caused by the soft iterations. 
\gls{ae}-4-\gls{sc} ($\mathbin{
\tikz[baseline]{\draw[color=matlab5,solid,thick] (0pt,.5ex) -- (2ex,.5ex) -- (4ex,.5ex);
\node[diamond,draw,solid, color=matlab5,thick, aspect=0.5, text width=4ex, inner sep=-5pt] (d) at (2ex,.5ex) {};
}}$) has similar performance with similar complexity.
CA-SCL-4 of 5G's shortened polar code has similar performance with the designed shortened polar code under SCL-4. 
We estimate $\Ebb(E_{13})=0.85\%$.
Hence, picking  automorphisms $\pi\in\Acal(\Ccalm)\leq\Pi$ require pre-computations, the adjusted \gls{ae} decoder takes into account the dynamic freezing constraints induced by selecting $\pi\in\Pi$. 
Adjusted AE-4-SCAN-5 ($\mathbin{
\tikz[baseline]{\draw[color=matlab1,dotted,thick] (0pt,.5ex) -- (2ex,.5ex) -- (4ex,.5ex);
\draw[solid, color=matlab1,thick](2ex,.5ex) circle (.5ex);}}$) shows similar performance with respect to AE-4-SCAN-5.
To complete the analysis, the shortened polar code $(115,51,\Zcalb)$ with the construction $\Zcal\in\Ical$ is decoded.
The $(128,64)$ mother code is $\mathcal{R}(3,7)$, its automorphism group is $\BLTA(7)$, 
and 4 automorphisms $\pi\in\BLTA(7)$ are picked randomly. The decoding performance is worse for \gls{ae}-4-\gls{bp}-200 ($\mathbin{
\tikz[baseline]{\draw[color=orange,dashed,thick] (0pt,.5ex) -- (2ex,.5ex) -- (4ex,.5ex);
\draw[solid, color=orange,thick] (1.5ex,0ex) -- (2.5ex,1ex);
\draw[solid, color=orange,thick] (1.5ex,1ex) -- (2.5ex,0ex);
}}$) and \gls{ae}-4-\gls{scan}-5 ($\mathbin{
\tikz[baseline]{\draw[color=matlab1,dashed,thick] (0pt,.5ex) -- (2ex,.5ex) -- (4ex,.5ex);
\draw[solid, color=matlab1,thick](2ex,.5ex) circle (.5ex);}}$) but better for \gls{ae}-4-\gls{sc} ($\mathbin{
\tikz[baseline]{\draw[color=matlab5,dashed,thick] (0pt,.5ex) -- (2ex,.5ex) -- (4ex,.5ex);
\node[diamond,draw,solid, color=matlab5,thick, aspect=0.5, text width=4ex, inner sep=-5pt] (d) at (2ex,.5ex) {};
}}$).

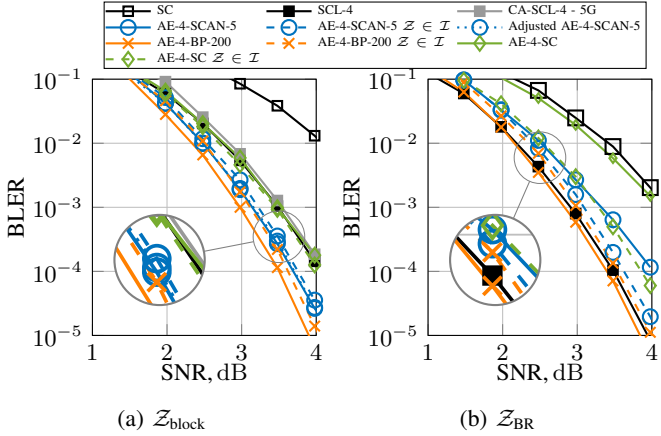
\begin{figure}[t]
    \centering
   \begin{subfigure}{0.49\columnwidth}
    \input{figures/ae_bp_error_perf.tex}
    \caption{$\Zcalb$}\label{fig:AE_block}
   \end{subfigure}
   \begin{subfigure}{0.49\columnwidth}
    \input{figures/ae_br.tex}
    \caption{$\Zcalbr$}\label{fig:AE_br}
   \end{subfigure}
  \caption{\gls{bler} of $(115,51,\Zcal)$ shortened polar codes.}\label{fig:error_AE}
  \end{figure}
\autoref{fig:AE_br} depicts the \gls{bler} performance for the \gls{br} shortened polar code $(115,51,\Zcalbr)$.
\Gls{sc} and \gls{scl}-4 are shown to perform better.
$\Abf_{\text{BR}}$ has $\NBstar=25$ as for $\Abf_{\text{block}}$ but has a worse shape \cite{AE_v1}.
Hence, the gain of \gls{ae}-4-\gls{sc} with respect to SC is tiny.
The \gls{bp} schedule permits to scramble more along the decoding such that \gls{ae}-4-\gls{bp}-200 is showing great performance. 
The shortening construction $\Zcal\subset\Ical$ allows to perform the decoding on the same mother code from \autoref{fig:AE_block}, i.e., $\mathcal{R}(3,7)$.
It allows to choose a larger pool of automorphisms with respect to those generated by $\Abf_{\text{BR}}$. 
With their schedule based on \gls{sc}, \gls{ae}-4-\gls{scan}-5 and \gls{ae}-4-\gls{sc} have improved performance, while that of \gls{ae}-4-\gls{bp}-20 is reduced.

\autoref{fig:lat_S13} depicts the average execution time of \gls{ae}-4-\gls{bp}-200 to decode a $(115,51)$ shortened code.
The $4$ \gls{bp} decoders include an early-termination scheme\cite{BP_early_termination}.
The average maximum number of BP iterations performed is noted $\mathbb{E}[T_{max}]$.
Due to the parallel nature of \gls{ae}-\gls{bp}, $\mathbb{E}[T_{max}]$ is used to compute the average execution time $\mathcal{L}_{\text{avg}}$ (in clock cycles) of \gls{ae}-\gls{bp} \cite{geiselhart_subcode} as $\mathcal{L}_{\text{avg}}=(2n+2)\mathbb{E}[T_{max}]+1$.
The $\Zcal\subseteq\Fcal$ shortening construction permits a more reliable feedback than that of $\Zcal\subset\Ical$. For both patterns, using $\Zcal\subset\Ical$ leads to worse average execution time.
The execution time of \gls{ae}-4-\gls{bp}-200 is compared with the decoding latency of \gls{scl}-4 with parallel processing \cite{SCL_LLR}    $\mathcal{L}_{\text{SCL}}=2N+K$.
$\mathcal{L}_{\text{SCL}}$ is matched at $\text{BLER}=6\cdot10^{-2}$ and at $\text{BLER}=7\cdot10^{-3}$ in the best and  worst case  scenario.
\begin{figure}
    \centering    \input{figures/iterations_bp_S13}
    \caption{Average execution time to decode $(115,51)$ code.}
    \label{fig:lat_S13}
\end{figure}
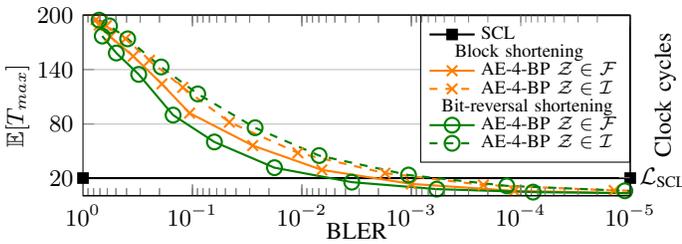
\section{Conclusions}\label{sec:conclusion}
In this paper, low-latency \gls{ae} decoding of shortened polar codes is proposed. It requires the knowledge of the affine automorphism group $\Acal(\Ccalm)$  of the mother code $\Ccalm$.
By describing $\Ccalm$ as a $\GN$-coset code,  $\Acal(\Ccalm)$ is proved to be not empty but has an irregular shape due to the shortening constraint $\Zcal\subseteq\Fcal$.
The shortening constraint $\Zcal\subset\Ical$ is also simulated leading to another $\GN$-coset code $\Ccalm$.
For $(115,51)$ shortened polar codes, the \gls{bler} performance is maximised under \gls{ae}-\gls{bp} and the construction $\Zcal\subseteq\Fcal$.
\gls{ae}-4-\gls{bp} outperforms \gls{scl}-4 by $0.5$ dB at $\text{BLER}=10^{-3}$ while having a lower average decoding execution time for $\text{BLER}=2\cdot10^{-3}$ and lower.
\vspace{-10pt}
\bibliographystyle{IEEEbib}
\bibliography{IEEEabrv,ConfAbrv,references}
\end{document}

%% file: figures/gn_coset_codes.tex

\begin{tikzpicture}
\usetikzlibrary{shapes.misc}

\tikzset{cross/.style={cross out, draw=black, minimum size=2*(#1-\pgflinewidth), inner sep=0pt, outer sep=0pt},
cross/.default={3pt}}
\draw[draw=black,fill=gray!10] (0,0) ellipse (4cm and 1.6cm);
\draw[draw=black,dashed,fill=gray!35] (0,0) ellipse (3.5cm and 1.25cm);
\draw[draw=black,dashed,fill=gray!35] (0,0) ellipse (2.65cm and 0.8cm);
\draw[draw=black,dashed,fill=gray!10] (0,0) ellipse (2.65cm and 0.8cm);
\draw[draw=black,fill=gray!60] (0,0) ellipse (1.5cm and 0.65cm);
\draw[draw=black,fill=gray!80] (0,0) ellipse (0.75cm and 0.375cm);
\draw (0,.30) node[cross,rotate=0] {};
\draw (0.45,0.05) node[cross,rotate=30] {};
\draw (-.45,0.25) node[cross,rotate=30] {};
\draw (-.2,0.05) node[cross,rotate=50] {};
\draw (.2,-0.25) node[cross,rotate=70] {};
\draw (-.21,-0.25) node[cross,rotate=90] {};

\draw (0,-1.6) node[below] {$\GN$-coset codes};
\draw (0,-0.92) node[] {$\GN$-coset mother code};
\draw (2.5,1.75) node[right] { Reed-Muller codes}; 
\draw[->] (2.5,1.75) -- (0.5,0.1); 
\draw (-2.5,1.75) node[left] { Polar codes}; 
\draw[->] (-2.5,1.75) -- (-0.45,-0.1); 
\draw (0,1.8) node[above] { Polar-like codes}; 
\draw[->] (0,1.8) -- (0.0,0.525); 
\end{tikzpicture}


%% file: figures/ae_bp_error_perf.tex
\usetikzlibrary{spy}
\begin{tikzpicture}[spy using outlines={circle, magnification=1.75, connect spies}]
  \pgfplotsset{
   label style = {font=\fontsize{9pt}{7.2}\selectfont},
   tick label style = {font=\fontsize{9pt}{7.2}\selectfont}
  }

  \begin{semilogyaxis}[%
    width=1.05\columnwidth,
    height=1.15\columnwidth,
    xmin=1, xmax=4,
    xtick={},
        ytick={0.1,0.01,0.001,0.0001,0.00001},
    xlabel={SNR,\,$\mathrm{dB}$},
    xlabel style={yshift=0.8em},
    ymin=1e-5, ymax=0.1,
    ylabel style={yshift=-0.6em},
    ylabel={BLER},
    yminorticks, xmajorgrids,
    ymajorgrids, yminorgrids,
    legend style={at={(0.05,1.04)},anchor=south west},
    legend style={legend columns=3, font=\tiny, row sep=-1mm},
    legend style={fill=white, fill opacity=1, draw opacity=1,text opacity=1}, 
    legend style={inner xsep=0.2pt, inner ysep=-1pt}, 
    legend cell align={left}, 
    mark size=1.6pt, mark options=solid,
    ]
\addplot[solid,color=black, mark=square, line width=0.8pt, mark size=1.8pt]
table[row sep=crcr]{
-0.520977 0.88968 \\
-0.0209767 0.764526\\
0.479023 0.706215 \\
0.979023 0.536481 \\
1.47902 0.420875 \\
1.97902 0.289687 \\
2.47902 0.173732 \\
2.97902 0.0849185\\
3.47902 0.0384734\\
3.97902 0.0131 \\
4.47902 0.00353447 \\
4.97902 0.000681193\\
};
\addlegendentry{SC}   
\addplot[solid,color=black, mark=square*, line width=0.8pt, mark size=1.8pt]
table[row sep=crcr]{
-0.520977 0.645995\\
-0.0209767 0.548246\\
0.479023 0.366032\\
0.979023 0.255363\\
1.47902 0.13587\\
1.97902 0.0598229\\
2.47902 0.0204349\\
2.97902 0.00540821\\
3.47902 0.00104529\\
3.97902 0.00014356\\
4.47902 1.2e-5 \\
};
\addlegendentry{SCL-4 }
  
\addplot[solid,color=black!40, mark=square*, line width=0.8pt, mark size=1.8pt]
table[row sep=crcr]{
-0.520977 0.788644\\
-0.0209767 0.679348 \\
0.479023 0.49505 \\
0.979023 0.317662 \\
1.47902 0.165782 \\
1.97902 0.0915416 \\
2.47902 0.0255781 \\
2.97902 0.00684481 \\
3.47902 0.00127654 \\
3.97902 0.000185 \\
4.47902 1.2e-05 \\
};
 \addlegendentry{CA-SCL-4 - 5G}

\addplot[solid,color=matlab1, mark=o, line width=0.8pt, mark size=2.8pt]
table[row sep=crcr]{
-0.520977 0.672043 \\
-0.0209767 0.5896\\
0.479023 0.372578 \\
0.979023 0.239234\\
1.47902 0.129333 \\
1.97902 0.0419674\\
2.47902 0.0100733 \\
2.97902 0.0018859\\
3.47902 0.000270281\\
3.97902 2.64e-05 \\
4.47902 4e-06 \\
};
    \addlegendentry{AE-4-SCAN-5 }
\addplot[dashed,color=matlab1,mark=o, line width=0.8pt, mark size=2.8pt]
table[row sep=crcr]{-0.520977 0.748503 \\
-0.0209767 0.5980\\
0.479023 0.41876 \\
0.979023 0.255624\\
1.47902 0.135355 \\
1.97902 0.0517063\\
2.47902 0.0117841\\
2.97902 0.0026823\\
3.47902 0.000356458\\
3.97902 3.56e-05\\ 
4.47902 2.8e-06 \\};
\addlegendentry{AE-4-SCAN-5 $\Zcal\in\Ical$}

\addplot[dotted,color=matlab1,mark=o, line width=0.8pt, mark size=2.8pt]
table[row sep=crcr]{
-0.520977 0.7062\\
-0.0209767 0.566\\
0.479023 0.37936\\
0.979023 0.22381\\
1.47902 0.117481\\
1.97902 0.042030\\
2.47902 0.009905\\
2.97902 0.001973\\
3.47902 0.000297\\
3.97902 2.74e-05\\
4.47902 2e-6\\
};
\addlegendentry{Adjusted AE-4-SCAN-5 }
    

    \addplot[solid,color=orange, mark=x, line width=0.8pt, mark size=2.8pt]
table[row sep=crcr]{
-0.520977 0.710227\\
-0.0209767 0.563063\\
0.479023 0.347705 \\
0.979023 0.190404\\
1.47902 0.106157 \\
1.97902 0.02815\\
2.47902 0.00657358\\
2.97902 0.00100025\\
3.47902 0.000114579\\
3.97902 6.4e-06\\
};
    \addlegendentry{AE-4-BP-200}

\addplot[dashed,color=orange, mark=x, line width=0.8pt, mark size=2.8pt]
table[row sep=crcr]{
-0.520977 0.759878 \\
-0.0209767 0.612745\\
0.479023 0.408497 \\
0.979023 0.247036\\
1.47902 0.122669 \\
1.97902 0.0459812 \\
2.47902 0.0108729 \\
2.97902 0.00174222\\ 
3.47902 0.000219728 \\
3.97902 1.4e-05 \\
4.47902 8e-07 \\
};
\addlegendentry{AE-4-BP-200 $\Zcal\in\Ical$}

\addplot[solid,color=matlab5, mark=diamond, line width=0.8pt, mark size=2.8pt]
table[row sep=crcr]{
-0.520977 0.67750\\
-0.0209767 0.5668\\
0.479023 0.382263\\
0.979023 0.259875\\
1.47902 0.15024\\
1.97902 0.0651551\\
2.47902 0.0203004\\
2.97902 0.0059103\\
3.47902 0.0009885\\
3.97902 0.0001790\\
4.47902 1.88e-05 \\
};
\addlegendentry{AE-4-SC}
\addplot[dashed,color=matlab5, mark=diamond, line width=0.8pt, mark size=2.8pt]
table[row sep=crcr]{
-0.520977 0.690608\\
-0.0209767 0.564\\
0.479023 0.388199 \\
0.979023 0.23786\\
1.47902 0.139821 \\
1.97902 0.059297\\
2.47902 0.0181449 \\
2.97902 0.004602\\
3.47902 0.00092854\\
3.97902 0.000126\\
4.47902 1.24e-05 \\
4.97902 8e-07 \\
};
\addlegendentry{AE-4-SC $\Zcal\in\Ical$}
     \coordinate (spypoint) at (axis cs:3.5,0.00035); 
  \coordinate (magnifyglass) at (axis cs:1.9,0.000152); 
  \end{semilogyaxis}

  \spy [gray, size=1.2cm] on (spypoint)
   in node[fill=white] at (magnifyglass);

\end{tikzpicture}%

    

%% file: figures/ae_br.tex
\usetikzlibrary{spy}
\begin{tikzpicture}[spy using outlines={circle, magnification=1.75, connect spies}]
  \pgfplotsset{
   label style = {font=\fontsize{9pt}{7.2}\selectfont},
   tick label style = {font=\fontsize{9pt}{7.2}\selectfont}
  }

  \begin{semilogyaxis}[%
    width=1.05\columnwidth,
    height=1.15\columnwidth,
    xmin=1, xmax=4,
    xtick={},
    ytick={0.1,0.01,0.001,0.0001,0.00001},
    xlabel={SNR,\,$\mathrm{dB}$},
    xlabel style={yshift=0.8em},
    ymin=1e-5, ymax=0.1,
    ylabel style={yshift=-0.6em},
    ylabel={BLER},
    yminorticks, xmajorgrids,
    ymajorgrids, yminorgrids,
    legend style={at={(0,0)},anchor=south west},
    legend style={legend columns=3, font=\small, row sep=-1mm},
    legend style={fill=white, fill opacity=1, draw opacity=1,text opacity=1}, 
    legend style={inner xsep=0.5pt, inner ysep=-1pt}, 
    legend cell align={left}, 
    mark size=1.6pt, mark options=solid,
    ]
\addplot[solid,color=black, mark=square, line width=0.8pt, mark size=2.8pt]
table[row sep=crcr]{
-0.520977 0.806452 \\
-0.0209767 0.7122\\
0.479023 0.548246\\
0.979023 0.398724\\
1.47902 0.238777 \\
1.97902 0.127097 \\
2.47902 0.0668271\\
2.97902 0.0248731\\
3.47902 0.00888099 \\
3.97902 0.0020303\\
4.47902 0.000374703\\
4.97902 6e-05\\
};
\addplot[solid,color=black, mark=square*, line width=0.8pt, mark size=1.8pt]
table[row sep=crcr]{
-0.520977 0.573394\\
-0.0209767 0.411\\
0.479023 0.257998 \\
0.979023 0.12853\\
1.47902 0.0597943 \\
1.97902 0.018033\\
2.47902 0.00431742\\
2.97902 0.000798\\
3.47902 0.0001045 \\
3.97902 8.4e-06 \\
};
 

\addplot[solid,color=matlab1, mark=o, line width=0.8pt, mark size=2.8pt]
table[row sep=crcr]{
-0.520977 0.6811\\
-0.0209767 0.5252\\
0.479023 0.3591 \\
0.979023 0.1971\\
1.47902 0.0963 \\
1.97902 0.03333\\
2.47902 0.01109 \\
2.97902 0.002690\\ 
3.47902 0.0006394\\
3.97902 0.000116\\
};
\addplot[dashed,color=matlab1,mark=o, line width=0.8pt, mark size=2.8pt]
table[row sep=crcr]{-0.520977 0.698324\\
-0.0209767 0.577\\
0.479023 0.384025 \\
0.979023 0.21079\\
1.47902 0.0977708 \\
1.97902 0.032778\\
2.47902 0.00830317\\
2.97902 0.001581\\
3.47902 0.00019781\\
3.97902 1.96e-05\\
4.47902 1.2e-06\\
};


    \addplot[solid,color=orange, mark=x, line width=0.8pt, mark size=2.8pt]
table[row sep=crcr]{
-0.520977 0.66666 \\
-0.0209767 0.4960\\
0.479023 0.3097\\
0.979023 0.1502\\
1.47902 0.063 \\
1.97902 0.017799\\
2.47902 0.00350292\\
2.97902 0.000585272\\
3.47902 7.1e-5 \\
3.97902 4.8e-6\\
};

\addplot[dashed,color=orange, mark=x, line width=0.8pt, mark size=2.8pt]
table[row sep=crcr]{
-0.520977 0.712251 \\
-0.0209767 0.5720\\
0.479023 0.390625\\
0.979023 0.193199\\
1.47902 0.0894454\\
1.97902 0.0267465\\
2.47902 0.00696631 \\
2.97902 0.0010616\\
3.47902 0.000134406\\
3.97902 1.12e-05\\ 
4.47902 8e-07 \\
};

\addplot[solid,color=matlab5, mark=diamond, line width=0.8pt, mark size=1.8pt]
table[row sep=crcr]{
-0.520977 0.748503\\
-0.0209767 0.631\\
0.479023 0.511247 \\
0.979023 0.33112\\
1.47902 0.219106 \\
1.97902 0.105485\\
2.47902 0.0506997 \\
2.97902 0.019325\\
3.47902 0.00583826\\
3.97902 0.001499\\
4.47902 0.00025545\\
};
\addplot[dashed,color=matlab5, mark=diamond, line width=0.8pt, mark size=2.8pt]
table[row sep=crcr]{
-0.520977 0.66313 \\
-0.0209767 0.543\\
0.479023 0.37092\\
0.979023 0.19936\\
1.47902 0.0961169 \\
1.97902 0.041308\\
2.47902 0.0121054 \\
2.97902 0.003010\\
3.47902 0.00049602\\
3.97902 6e-5\\
};
     \coordinate (spypoint) at (axis cs:2.5,0.006); 
  \coordinate (magnifyglass) at (axis cs:1.9,0.000152); 
  \end{semilogyaxis}

  \spy [gray, size=1.2cm] on (spypoint)
   in node[fill=white] at (magnifyglass);

\end{tikzpicture}%

    

%% file: figures/iterations_bp_S13.tex
\begin{tikzpicture}
    \def\width{\columnwidth}
    \def\height{1.3\columnwidth}
\pgfplotsset{width=7cm,compat=1.3}
  \pgfplotsset{
    label style = {font=\fontsize{9pt}{7.2}\selectfont},
    tick label style = {font=\fontsize{9pt}{7.2}\selectfont}
  }
   \begin{axis}[%
    width=\columnwidth,
    height=0.45\columnwidth,
    xmin=1e-5, xmax=1,
    xlabel={\gls{bler}},
    xlabel style={yshift=1.6em},
    ymin=1, ymax=200,
    x dir=reverse,
    xmode=log,
    ylabel style={yshift=-0.4em},
    ylabel={$\mathbb{E}[T_{max}]$},
    ytick={20,80,140,200},
    xlabel style={yshift=-0.8em},
    ymajorgrids, yminorgrids,
    legend style={at={(1.0,1.0)},anchor=north east},
    legend style={legend columns=1, font=\scriptsize, column sep=0mm, row sep=-1.2mm}, 
    legend style={fill=white, fill opacity=0.75, draw opacity=0.75,text opacity=1}, 
    legend cell align={left}, 
    mark size=2.8pt, mark options=solid,
    ] 
        \addplot[color=black,  mark=square*,mark size=1.8pt, line width=0.8pt] table[row sep=crcr]{%
    1 20.5\\    
    1e-5 20.5\\    
    };
    \addlegendentry{SCL}
\addlegendimage{empty legend}
    \addlegendentry{\hspace{-10pt}Block shortening}
    \addplot[color=orange,  solid,mark=x, line width=0.8pt] table[row sep=crcr]{%
    0.710227	189.074\\
0.563063	176.282\\
0.347705	154.694\\
0.190404	124.183\\
0.106157	92.0815\\
0.0281500	56.3295\\
0.00657358	29.5220\\
0.00100025	14.0853\\
0.000114579	6.44130\\
6.40000e-06	3.34629\\
    };
    \addlegendentry{AE-4-BP $\Zcal\in\Fcal$}

        \addplot[color=orange,  dashed,mark=x, line width=0.8pt] table[row sep=crcr]{%
0.759878 195.158\\
0.612745 188.304\\
0.408497 174.057\\
0.247036 150.149\\
0.122669 120.381\\
0.0459812 82.101\\
0.0108729 48.2638\\
0.00174222 25.9433\\
0.000219728 12.9119\\
1.4e-05 6.88982\\
8e-07 4.33035\\
    };
    \addlegendentry{AE-4-BP $\Zcal\in\Ical$}

\addlegendimage{empty legend}
    \addlegendentry{\hspace{-15pt}Bit-reversal shortening}
        \addplot[color=green!50!black,  solid,mark=o, line width=0.8pt] table[row sep=crcr]{%
0.66666 176.92\\
0.4960 158.24\\
0.3097 134.70\\
0.1502 89.87\\
0.063 60.28\\
0.017799 31.94\\
0.00350292 16.02\\
0.000585272 8.34\\
7.72e-5 5.05\\
4.8e-6 3.37\\
    };
    \addlegendentry{AE-4-BP $\Zcal\in\Fcal$}
        \addplot[color=green!50!black,  dashed,mark=o, line width=0.8pt] table[row sep=crcr]{%
0.712251 194.45\\
0.572082 188.121\\
0.390625 173.864\\
0.193199 142.875\\
0.0894454 113.228\\
0.0267465 75.9483\\
0.00696631 45.3064\\
0.00106162 23.8452\\
0.000134406 11.9679\\
1.12e-05 6.50643\\
8e-07  4.19861\\
    };
    \addlegendentry{AE-4-BP $\Zcal\in\Ical$}

\end{axis}    
 \begin{axis}[%
    width=\columnwidth,
    height=0.45\columnwidth,
    xmin=1e-5, xmax=1,
    ymin=1, ymax=3200,
    x dir=reverse,
    xmode=log,
    xtick={},
    xticklabels={},
    ytick={320},
    yticklabels={$\mathcal{L}_{\text{SCL}}$},
    ylabel style={yshift=2em},
      axis y line*=right,
    ylabel style={yshift=-0.1em, xshift=0.4em},
    ylabel={Clock cycles},
    yminorticks,
    mark size=1.8pt, mark options=solid,
    ] 
\end{axis}    
\end{tikzpicture}%